\documentclass[11pt,a4paper]{article}
\pdfoutput=1
\usepackage{graphicx}
\usepackage{jheppub}

\usepackage{amsmath,amssymb,mathtools}
\usepackage{amsthm}
\usepackage{nccmath}
\usepackage[compatibility=false]{caption} 
\usepackage{subcaption} 
\usepackage{booktabs}
\usepackage{multirow}

\newtheorem{thm}{Theorem}%[section]
\newtheorem{conj}{Conjecture}%[section]

\def\beq{\begin{equation}}
\def\eeq{\end{equation}}
\def\bea{\begin{eqnarray}}
\def\eea{\end{eqnarray}}
\def\beu{\begin{quotation}}
\def\eeu{\end{quotation}}

\begin{document}

\title{On the Causal Set-Continuum Correspondence}

\author[a,b]{Mehdi Saravani,}
\author[a,b]{Siavash Aslanbeigi}
\affiliation[a]{Perimeter Institute for Theoretical Physics, 31 Caroline St. N., Waterloo, ON, N2L 2Y5, Canada}
\affiliation[b]{Department of Physics and Astronomy, University of Waterloo, Waterloo, ON, N2L 3G1, Canada}
\emailAdd{msaravani@perimeterinstitute.ca}
\emailAdd{saslanbeigi@perimeterinstitute.ca}

\abstract{
We present two results which concern certain aspects of the question: when is a causal set well approximated by a Lorentzian manifold?  The first result is a theorem which shows that the number-volume correspondence, if required to hold even for arbitrarily small regions, is best realized via 
Poisson sprinkling. The second result concerns a family of lattices in $1+1$ dimensional Minkowski space, known as Lorentzian lattices, 
which we show provide a much better number-volume correspondence than Poisson sprinkling for large volumes. We argue, however, that this feature should not persist 
in higher dimensions. We conclude by conjecturing a form of the aforementioned theorem that holds under weaker assumptions, namely that
Poisson sprinkling provides the best number-volume correspondence in $3+1$ dimensions for spacetime regions with macroscopically large volumes.
}

\maketitle
\flushbottom

%---------------------------------------------------------------------------------------------------------------------------------
%---------------------------------------------------------------------------------------------------------------------------------
%----------------------------------------------------------Introduction-------------------------------------------------------
%---------------------------------------------------------------------------------------------------------------------------------
%---------------------------------------------------------------------------------------------------------------------------------
\section{Background}
From the viewpoint of causal set theory, the continuum spacetime of general relativity is only fundamental to 
the extent that it provides a good approximation to an underlying causal set \cite{Sorkin_1,Sorkin_2,Sorkin_3,Dowker_1,Henson}. 
Once a full dynamical theory of causal sets is available, it is necessary to judge whether or not the result of evolution looks anything like
the universe we observe at low energies. Therefore, criteria must be established to determine how well a Lorentzian geometry $(M,g)$ approximates a causal set $(C,\prec)$. 
\footnote{
A causal set (causet) is a set $C$ endowed with a binary relation $\prec$ such that for all $x, y, z \in  C$ the following axioms are satisfied:
(1) transitivity: $x\prec y$ $\&$ $y\prec z \Rightarrow x\prec z$, (2) irreflexivity: $x \nprec x$, (3):  local finiteness: $|\{y \in C|x \prec y \prec z\}| < \infty$. 
}
One natural criterion is to require the existence of an injective map $f:C\to M$ which preserves causal relations: $\forall$ $x,y\in C$, $x\prec y$ if and only if $f(x)\in J^{-}(y)$,
where $ J^{-}(y)$ is the set of all points in $M$ which lie in the causal past of $y$.
We would then say that $C$ is \emph{embeddable} in $M$.
Of course, it is not very likely for a causal set which has emerged out of the dynamics to be \emph{exactly} embeddable in \emph{any} spacetime. 
Close to the discreteness scale, for instance, one would expect the causal set to be fairly chaotic. 
Therefore, a certain degree of \emph{coarse graining} must be done before embedding is possible. 
It might also be necessary to introduce some notion of approximate embedding, because matching \emph{all} causal relations exactly (and there would be a lot of them)
seems too stringent a requirement. 
%\footnote{
%It is not very likely that a causal set which has emerged out of the dynamics would be \emph{exactly} embeddable in \emph{any} spacetime. 
%Close to the discreteness scale, for instance, one would expect the causal set to be fairly choatic. 
%Therefore, a certain degree of \emph{coarse graining} must be done before embedding is possible. 
%}
Once these issues are settled and embedding is possible, one last piece of information is required: scale. 
%As far as the causal set-continuum correspondence is concerned, preserving causal relations is not enough. 
 %because causal relations are left invariant under conformal tranformations. 
This is because preserving causal relations cannot distinguish between spacetimes
whose metrics are conformally related.
%If we do find such a map, we have also found a map to any spacetime whose metric is conformally related to $g$. 
%To break this degeneracy, information about scale is needed. 
Causal sets contain information about scale implicitly through counting of elements, because they are locally finite (i.e. discrete).
%they also contain information about scale, and not just causal relations.
To make use of this property, one also requires a \emph{number-volume (N-V) correspondence}:
the number $N_S$ of embedded points in any spacetime region $S\subset M$ should ``reflect" its volume $V_S$:
\beq
N_S\approx\rho V_S=\rho\int_{S}\sqrt{-g(x)}d^Dx,
\eeq
where $\rho$ is a constant, thought to be set by the Planck scale, which represents the number density of points. 
Of course, this correspondence cannot be \emph{exactly} true, 
%Of course, the number of embedded points in a spacetime region $S$ with volume $V_S$ cannot always be \emph{exactly} $\rho V_S$.
the most obvious reason being that $\rho V_S$ is not always an integer. 
%forcing us to demand instead for all regions $S$: 
%\beq
%\left |\frac{N_S-\rho V_S}{\rho V_S}\right|\ll 1.
%\eeq
Also, for any embedding, there would always be infinitely many empty regions meandering through the embedded points. 
These issues can be addressed by first settling on the types of ``test regions" $S$, and then requiring the correspondence in a statistical sense.
To do so, let us first note that 
%as far as macroscopic physics is concerned, 
the causal set-continuum correspondence is only physically meaningful on scales 
much larger than the discreteness scale. Therefore, $S$ should be a region whose spacetime volume is much larger than that set by the discreteness scale.
The shape of $S$ can be picked to disallow regions that meander through the embedded points but have large volumes. 
A natural choice, given that spacetime is Lorentzian, is the causal interval $I(x,y)$: given any two timelike points $x\prec y\in M$, $I(x,y)$ is the collection of
all points in the causal future of $x$ and the causal past of $y$.  
%Having decided on the types of test regions $S$, let us give a more precise meaning to 
%$N_S\approx\rho V_S$. It seems natural to require $|N_S-\rho V_S|\ll \rho V_S$ for all regions $S$. 
Having decided on the types of test regions, the number-volume correspondence can
be formulated as follows: pick at random $M$ causal intervals $S_1,S_2,\dots,S_M$ with the same volume $V\gg \rho^{-1}$, and let
$N_1,N_2,\dots,N_M$ be the number of embedded elements in these regions, respectively. We then require that as $M\to\infty$:
\beq
\langle N\rangle{=}\rho V, \qquad
\frac{\delta N}{\langle N\rangle}=\frac{\sqrt{\langle(N-\langle N\rangle)^2\rangle}}{\langle N\rangle}\ll 1.
\label{NVreq}
\eeq
Having the $N$-$V$ formulation at hand,
\footnote{
It may seem more natural to require instead $|N_S-\rho V_S|\ll \rho V_S$ for all test regions $S$. This requirement, however, is a bit too stringent.
Even if there is only one region which violates this condition, the $N$-$V$ correspondence would be rendered unsatisfied. Requiring \eqref{NVreq}
ensures that \emph{almost all} regions have volumes representative of the number of embedded points in them.
}
the key question becomes: what is the map that 
realizes the number-volume correspondence with the least noise?

The attitude in the causal set program is that this mapping is best done through \emph{Poisson sprinkling}.
In this approach, one first reverses direction by obtaining a causal set $C(M)$ from a given spacetime $(M,g)$: randomly select points from $M$ using 
the Poisson process at density $\rho$ and endow the selected points with their causal relations. The probability of selecting $n$ points from a region with volume $V$ is 
 \footnote{
The Poisson process can be obtained by dividing spacetime into small regions of volume $dV$
so that (i) in each infinitesimal region one point can be selected at most, and (ii) this selection
happens with the probability $\rho dV$ \textit{independent} of outside regions. Then, the probability of
selecting $n$ points in a volume $V$ is $P(n) = \binom{V/dV}{n}(\rho dV )^n(1-\rho dV )^{V/dV-n}$, which converges to \eqref{poisson} in 
the limit $dV \to 0$. 
}
\beq
P(n)=\frac{(\rho V)^ne^{-\rho V}}{n!}.
\label{poisson}
\eeq
Both the expectation value and variance of the number of selected points in a region with volume $V$ is equal to $\rho V$:
\beq
\langle N\rangle_{Pois}{=}\rho V, \qquad
\frac{\delta N_{Pois}}{\langle N\rangle_{Pois}}=\frac{1}{\sqrt{\rho V}}.
\eeq
%The number-volume requirement is now true on average and becomes better as $V$ gets larger.
% by $1/\sqrt{V}$.
The causal set-continuum correspondence is then judged as follows: \emph{a Lorenztian manifold $(M,g)$ is well-approximated by a causal set $C$ if and only if $C$ could have arisen from a sprinkling of $(M,g)$ with ``high probability"}. This definition is consistent with the $N$-$V$ requirement formulated above: if $C$ is embeddable as a ``large enough" sprinkling of $(M,g)$, 
\eqref{NVreq} would be satisfied because of the ergodic nature of the Poisson process. 
%\footnote{
%This feature is observed in Monte Carlo simulations of the Poisson process when the random variable under study is $N_S$. 
%Whether or not this can be cast as a theorem is a separate issue.
%}
The ``high probability" requirement is necessary to ensure that a large enough sprinkling is indeed obtained.
%given that the probability of obtaining $n$ points in a region with volume $V$ is only significant for $n\in(\rho V-\sqrt{\rho V},\rho V+\sqrt{\rho V})$.
Ultimately, one needs to decide how high ``high probability" is. A practical meaning could be that observables (such as dimension, proper time, etc)
are not too wildly far from their mean \cite{Henson}. It is interesting to note that \emph{any} embeddable $C$ has a finite probability of being realized through a Poisson sprinkling. 
This formulation of the causal set-continuum correspondence can be used for any point process (i.e. not just Poisson) which satisfies the $N$-$V$ requirement on average.
%\footnote{
%For a generic point process $\xi$, it might also be necessary to impose the condition that any embeddable causal set should be obtainable from realizations of $\xi$ with finite probability. 
%}
%In other words, once we a causal set $C$ is embeddable into a Lorenztian manifold $(M,g)$, we compute the probability of producing $C$ via a Poisson sprinkling of $(M,g)$. If this probability is ``high enough", we say that $(M,g)$ is well-approximated by $C$. 

Poisson sprinkling has many desirable features. 
It has been proven that even its realizations do not select a preferred frame in Minkowski space \cite{Bombelli}. 
If this mapping really does provide the best causal set-continuum dictionary, it is intriguing that Lorentz invariance should follow as a
biproduct. Also, Poisson sprinkling works in \emph{any} curved background. 
Even the extra requirement of the shape of test regions as causal intervals is not necessary in this context. 
On the way to proving that %the Hauptvermutung of causal set theory, namely 
the causal set structure is in principle rich enough to give rise to a smooth Lorentzian manifold, 
Poisson sprinkling has played a central role. But is it unique?

This paper contains two results which (we hope) shed some light on certain aspects of this question. 
The first result is that the number-volume correspondence, if required to hold even for arbitrarily small regions, is best realized via 
Poisson sprinkling.
The second result concerns a family of lattices in $1+1$-dimensional Minkowski space, known as Lorentzian lattices, 
which we show provide a better number-volume correspondence than Poisson sprinkling for large volumes. 
\footnote{
The existence of Lorentzian lattices in $1+1$-dimensional Minkowski space, and that they might be a contender for the Poisson process, was suggested by Aron Wall 
to Rafael Sorkin, who then mentioned it to us. 
}
We argue, however, that this feature should not persist 
in higher dimensions and that it is special to $1+1$-dimensional Lorentzian lattices. 
We conclude by conjecturing that
Poisson sprinkling provides the best number-volume correspondence in $3+1$ dimensions for spacetime regions with macroscopically large volumes.
%---------------------------------------------------------------------------------------------------------------------------------
%---------------------------------------------------------------------------------------------------------------------------------
%-----------------------------------Nothing Beats Poisson for Planckian Volumes}-------------------------------
%---------------------------------------------------------------------------------------------------------------------------------
%---------------------------------------------------------------------------------------------------------------------------------
\section{Nothing beats Poisson for Planckian volumes}
In this Section we prove that the number-volume correspondence is best realized via Poisson sprinkling for arbitrarily small volumes.
We set $\rho=1$ in the statement and proof of the theorem. 

\begin{thm}
Let $\xi$ be a point process whose realizations are points of a smooth Lorentzian manifold $(M,g)$. 
Let $N_S$ be the random variable which counts the number of points in a causal interval $S\subset M$: it takes on a value $n\in\{0,1,2,\dots\}$ with probability $P_S(n)$.
%where $P_S(n)$ is the probability of realizing $n$ points in $S$.
Assume also
that $\xi$ realizes the number-volume correspondence on average $\forall$ $S$: $\langle N_S\rangle=\sum_{n=0}^{\infty}nP_S(n)=V_S$, where
$V_S$ is the spacetime volume of $S$. 
Then, $\nexists$ $\xi$ such that $\forall$ $S$:
\beq
\langle(N_S-V_S)^2\rangle\le\alpha V_S \quad
\text{where} \qquad
0\le\alpha<1.
\label{poisThm}
\eeq
%Let $N_V$ be a discrete random variable which takes on a value $n\in\{0,1,2,\cdots\}$ with probability $P_V(n)$, and whose mean is 
%$V>0$: $\langle N_V\rangle=\sum_{n=0}^{\infty}nP_V(n)=V.$
%%\beq
%%\langle N_V\rangle=\sum_{n=0}^{\infty}nP_V(n)=V.
%%\eeq
%Then, $\nexists$ a point process such that $\forall$ $V>0$:
%\beq
%\langle(N_V-V)^2\rangle=\alpha V \quad
%\text{where} \qquad
%0\le\alpha<1.
%\label{poisThm}
%\eeq
\end{thm}
\begin{proof}
It is shown in Appendix \ref{saravProcess} that the variance of any random variable $N_S$ which takes on a value $n\in\{0,1,2,\dots\}$ with probability $P_S(n)$, and
whose mean is $V_S>0$,
must satisfy the inequality
\beq
\langle(N_S-V_S)^2\rangle \ge (V_S-n_*)(n_*+1-V_S),
\label{ineQ}
\eeq
where $n_*$ is the largest integer which is smaller than or equal to $V_S$. To see why this should be true, 
consider choosing $P_S(n)$ to obtain the least possible variance for $N_S$. Intuitively, this can be done by letting  
$P_S(n)=0$ $\forall$ $n\neq n_*,n_*+1$. Requiring $\langle N_S\rangle=V_S$ and $\sum_{n=0}^{\infty}P_S(n)=1$
then implies $P_S(n_*)=n_*+1-V_S$ and $P_S(n_*+1)=V_S-n_*$, which leads to the variance
$(V_S-n_*)(n_*+1-V_S)$. 
%
%Therefore, the variance of any point process $N_V$ must satisfy the inequality 
%\beq
%\langle(N_V-V)^2\rangle \ge (V-n_*)(n_*+1-V).
%\label{ineQ}
%\eeq
The formal proof of this result is given in Appendix \ref{saravProcess}. 

Let us now proceed to prove the theorem by contradiction. Assume there exists
$0\le\alpha<1$ such that $\langle(N_S-V_S)^2\rangle\le\alpha V_S$ for all $S$. It then follows from \eqref{ineQ} that
\beq
(V_S-n_*)(n_*+1-V_S)\le\alpha V_S \qquad \forall \qquad S.
\eeq
This, however, is clearly false because any region $S$ with $V_S<1-\alpha$ violates this condition.
\end{proof}

The proof of this theorem rests heavily on regions with Planckian volumes. For instance, had we required  
the condition \eqref{poisThm} for regions with $V_S>1$, the proof would not have gone through.
As we mentioned previously though, the causal set-continuum correspondence is only physically meaningful on scales 
much larger than the discreteness scale. 
In order to show that nothing \emph{really} beats Poisson, our result would have to be generalized to the case of larger volumes. 
We have, however, found a counter example to this conjecture in the case of $1+1$-dimensional Minkowski space. 
As we shall see in the next Section, 2D Lorentzian lattices realize the number-volume correspondence much better than Poisson sprinkling
for large volumes. 
%We argue, however, that this result is special to $2$D and that Lorentzian lattices in $4$D would not share this feature.

%---------------------------------------------------------------------------------------------------------------------------------
%---------------------------------------------------------------------------------------------------------------------------------
%--------------------------------------------------Lorentzian Lattices-----------------------------------------------------
%---------------------------------------------------------------------------------------------------------------------------------
%---------------------------------------------------------------------------------------------------------------------------------
\section{Lorentzian Lattices}

Why is a \emph{random}, as opposed to regular, embedding of points thought to provide the best number-volume correspondence?
Consider, for instance, a causal set which is embeddable as a regular lattice in $1+1$-dimensional Minkowski space. 
Our intuition from Euclidean geometry dictates that such a lattice should at least match, if not beat, a random sprinkling in uniformity. 
Why not, then, use a regular lattice as opposed to Poisson sprinkling? Figure \ref{fig:reglat} shows what goes wrong in Lorentzian signature.
Although the lattice is regular in one inertial frame, it is highly irregular for a boosted observer. 
Therefore, there are many empty regions with large volumes, which leads to a poor realization 
of the number-volume correspondence. 
\begin{figure}
%	\centering
	\begin{subfigure}[t]{0.52\hsize}
		\includegraphics[width=\hsize]{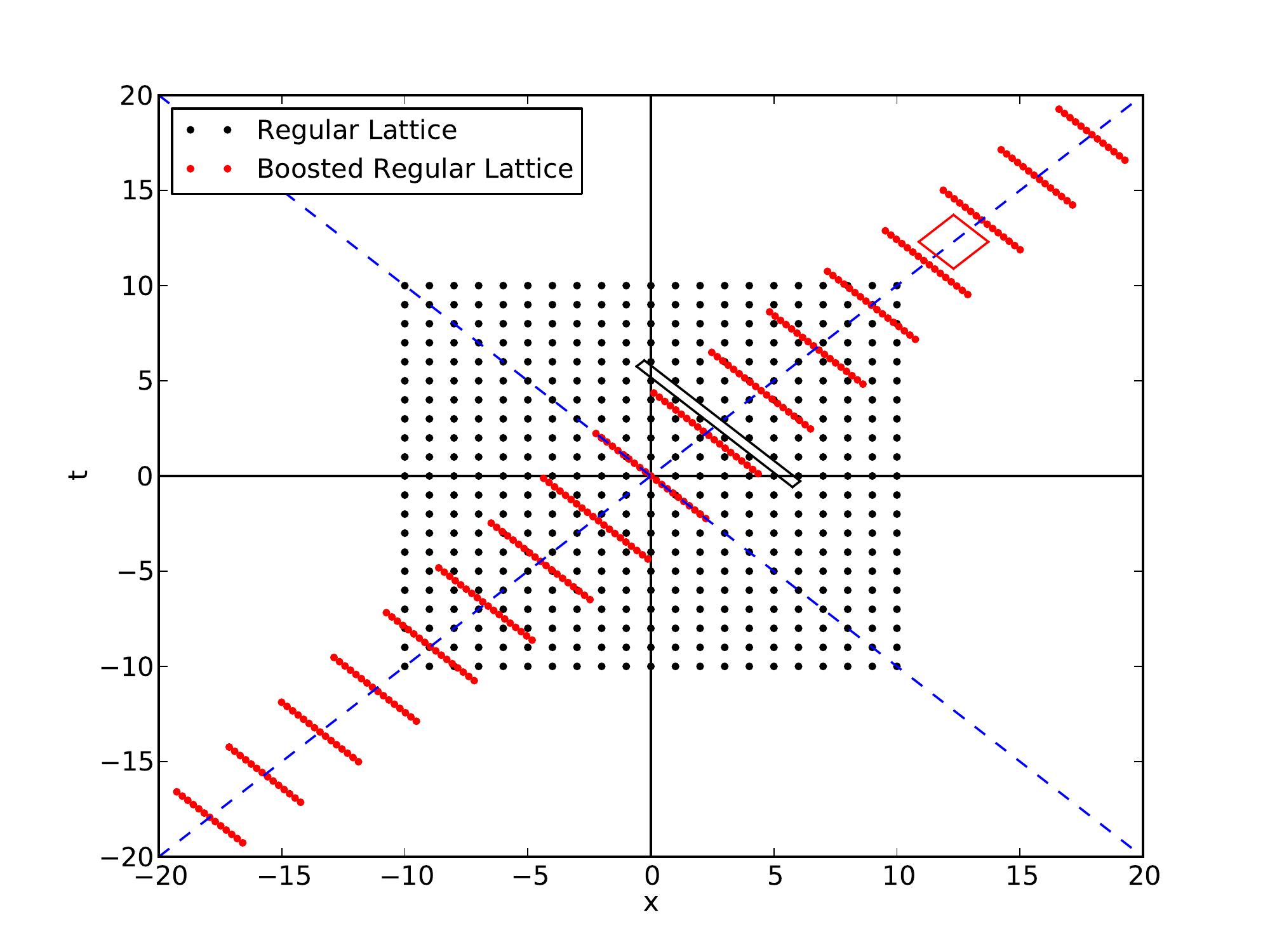}
		\caption{}
		\label{fig:reglat}
	\end{subfigure}%
	\begin{subfigure}[t]{0.52\hsize}
		\includegraphics[width=\hsize]{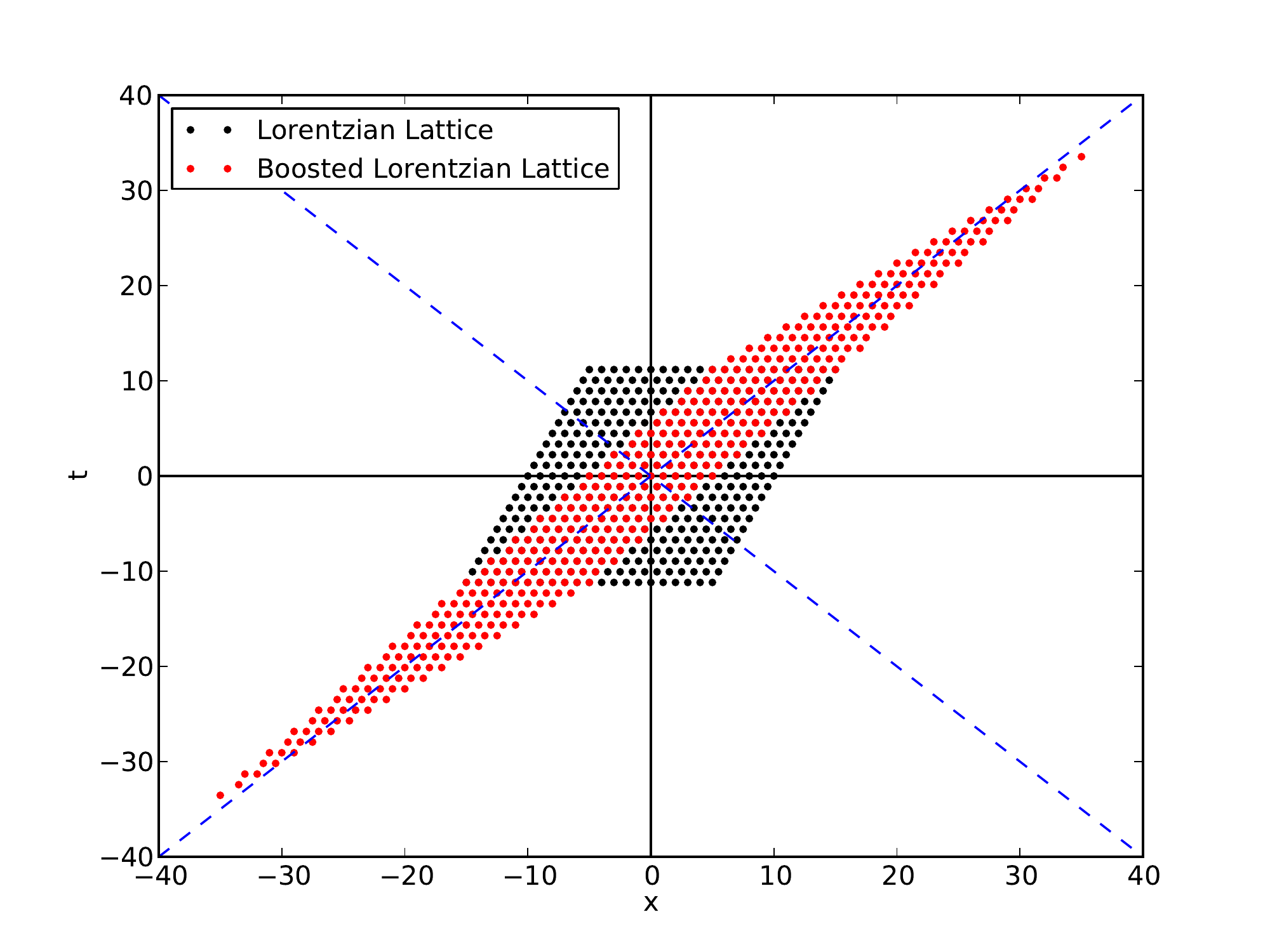}
		\caption{}
		\label{fig:lorlat}
	\end{subfigure}
	\caption{(a) The black dots show a lattice on the integers. The red dots are an active boost of this lattice  by velocity $v=\tanh(1.5)$. The red diamond is a causal interval in the boosted frame which contains no points. The black diamond is the same causal interval as seen in the original frame. (b) The black dots show a Lorentzian lattice generated by the timelike vector $\xi_{(0)}=(\sqrt{5}/2,1/2)$, and the spacelike vector $\xi_{(1)}=(0,1)$. The red dots are boosts of the Lorentzian lattice by $v=\sqrt{5}/3$, showing that this particular boost takes the lattice to itself.}
\end{figure}
%\begin{figure}
%	\centering
%	\includegraphics[width=0.7\hsize]{regLatNV}
%	\caption{The black dots show a lattice on the integers. The red dots show an active boost of this lattice by $v=\tanh(1.5)$. The red diamond shows a causal interval in the boosted frame which contains no points. The black diamond is the same causal interval as seen in the original frame.}
%	\label{fig:reglat}
%\end{figure}
Are there any regular lattices in $1+1$ that do not have this problem? As it turns out, the answer is yes: Lorentzian lattices. 
These are lattices which are invariant under a discrete subgroup of the Lorentz group. Such a lattice is shown in Figure \ref{fig:lorlat}:  it goes to itself 
under the action of a discrete set of boosts. We have classified all 2D Lorentzian lattices in Appendix \ref{LL2Ddetails}.
In the case of the integer lattice shown in Figure \ref{fig:reglat}, the more it is boosted, the more irregular it becomes. A Lorentzian lattice, however,
does not have this problem because it eventually goes to itself.
%as it is left invariant under a discrete set of boosts. 
It is then reasonable to expect a better number-volume correspondence in this case.
%from Lorentzian lattices.

%\begin{figure}
%	\centering
%	\includegraphics[width=0.7\hsize]{lorentzLat}
%	\caption{The black dots show a Lorentzian lattice generated by the timelike vector $\xi_{(0)}=<\sqrt{5}/2,1/2>$ and spacelike vector $\xi_{(1)}=<0,1>$. The red dots are boosts of the black dots, where $\phi=\cosh^{-1}(3/2)$.}
%	\label{fig:lorlat}
%\end{figure}	

We have investigated the $N$-$V$ correspondence for various Lorentzian lattices using simulations. Figure \ref{fig:LLNV} shows the result of one such analysis on the lattice shown in Figure\ref{fig:lorlat}. 
The setup is as follows: we consider $1000$ different causal diamonds with the same volume $V$, whose centres and shapes vary randomly throughout the lattice.
\footnote{
We made sure to include ``stretched out" causal diamonds, such as the black diamond shown in Figure \ref{fig:reglat}, as they are responsible for the poor realization of
the number-volume correspondence in the integer lattice.
}
For each realization, the number of lattice points inside the causal diamond is counted, leading to a
distribution of the number of points for a given volume $V$. This procedure is then repeated for different volumes. 
As it can be seen from Figure \ref{fig:LLNV}, the Lorentzian lattice shown in Figure \ref{fig:lorlat} realizes the
number-volume correspondence with much less noise than Poisson sprinkling for macroscopic volumes. In fact, Figure \ref{stat} shows that the dispersion about the mean is barely growing with volume at all. The same exercise with the integer lattice results in a huge dispersion, much larger than that of Poisson, which is to be expected. 
%\begin{figure}
%	\centering
%	\includegraphics[width=\hsize]{boostLL}
%	\caption{Various boosts of the Lorentzian lattice above.}
%	\label{fig:lorlat}
%\end{figure}

%\begin{figure}
%	\centering
%	\includegraphics[width=0.8\hsize]{NV}
%	\caption{Various boosts of the Lorentzian lattice above.}
%	\label{fig:lorlat}
%\end{figure}
%
%\begin{figure}
%	\centering
%	\includegraphics[width=0.8\hsize]{statHist}
%	\caption{Various boosts of the Lorentzian lattice above.}
%	\label{fig:lorlat}
%\end{figure}
\begin{figure}
	\centering
	\begin{subfigure}[t]{0.52\hsize}
		\includegraphics[width=\hsize]{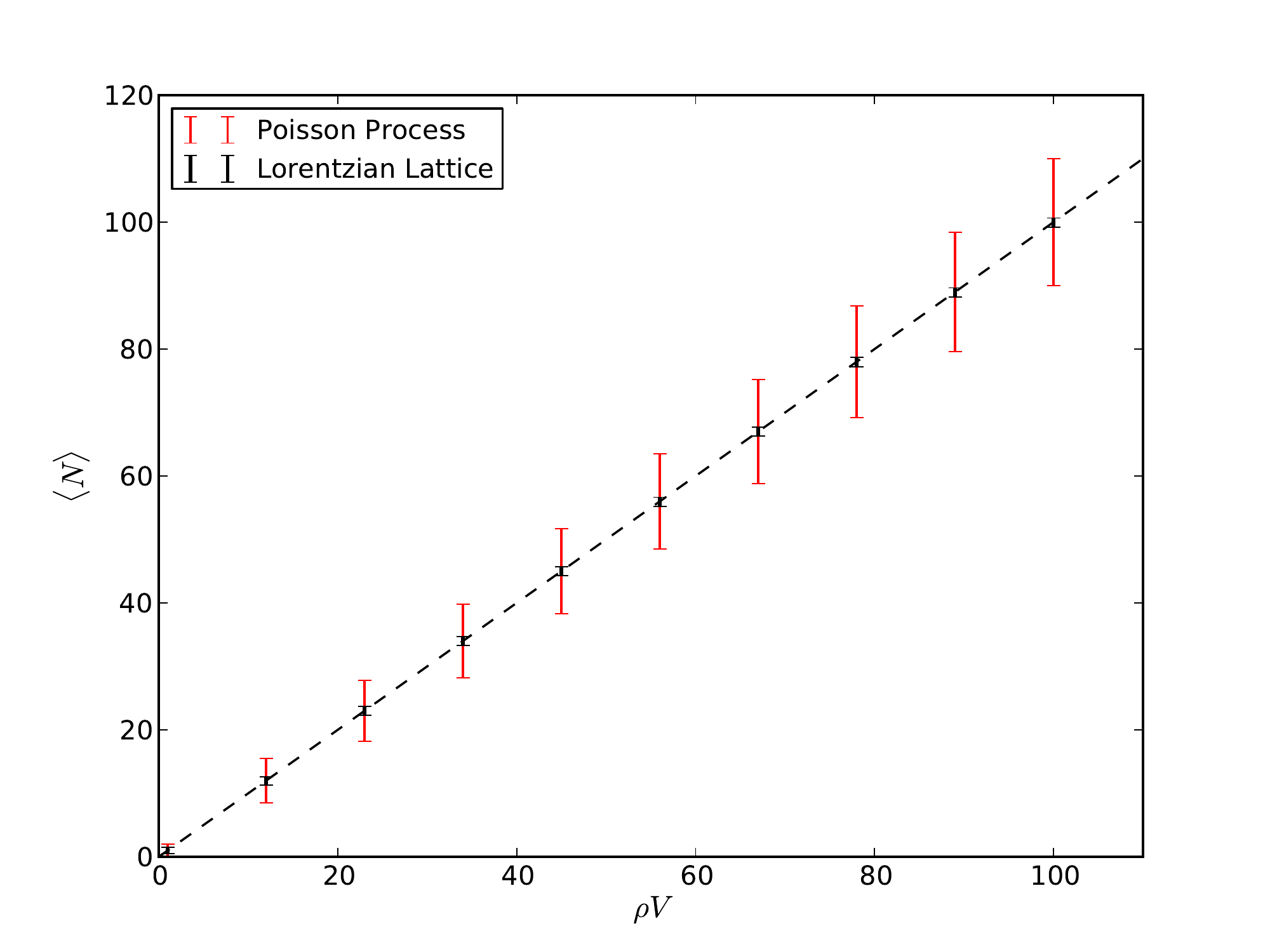}
		\caption{}
		\label{NV}
	\end{subfigure}%
	\begin{subfigure}[t]{0.52\hsize}
		\includegraphics[width=\hsize]{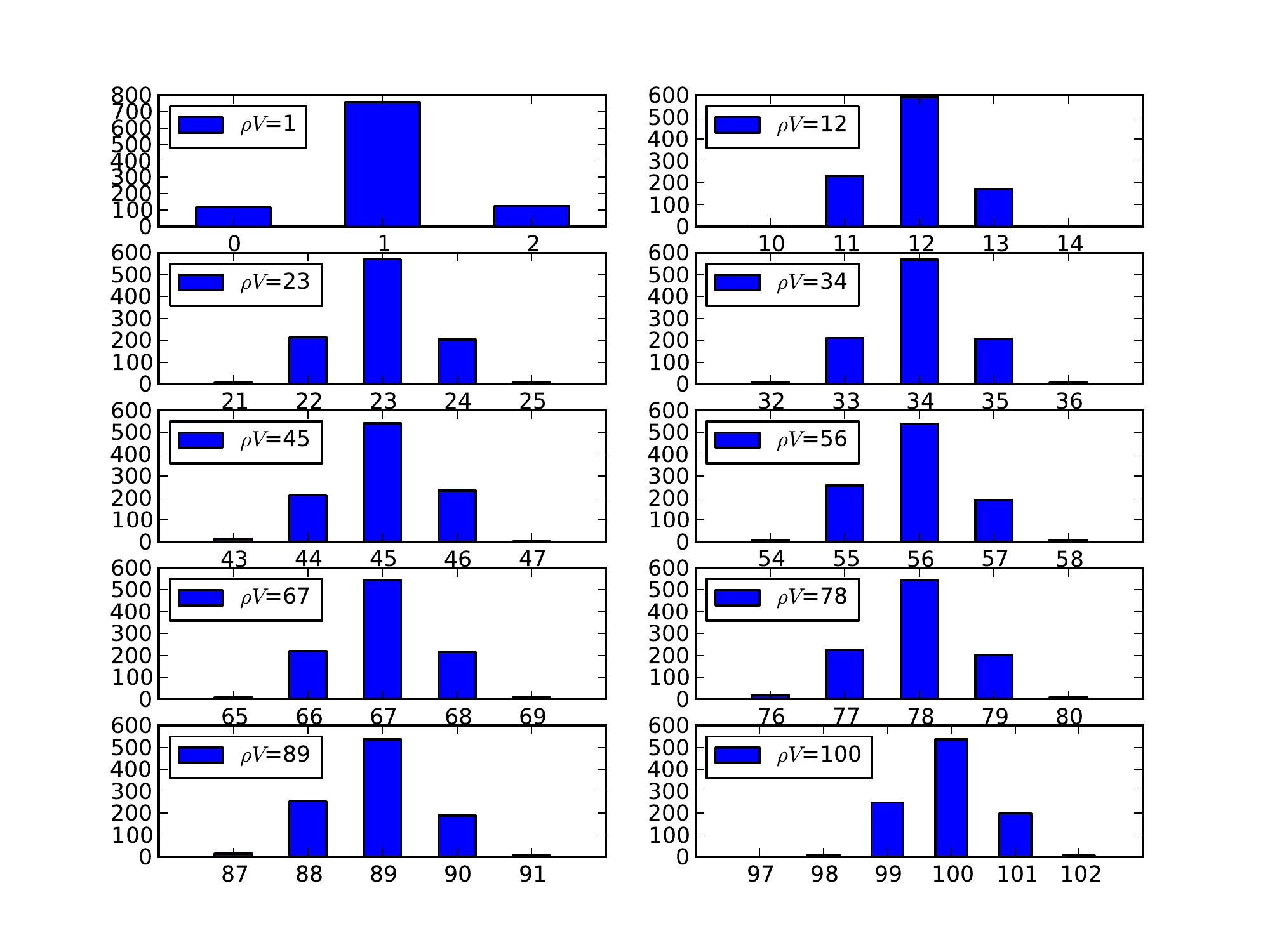}
		\caption{}
		\label{stat}
	\end{subfigure}
	\caption{The number-volume correspondence for the Lorentzian lattice shown in Figure \ref{fig:lorlat}. (a) The mean and standard deviation of the number of points. (b) The histogram of the number of points for different volumes.}
	\label{fig:LLNV}
\end{figure}

What about Lorentzian lattices in $3+1$ dimensions? Would they also realize the number-volume correspondence better than Poisson sprinkling? 
What is quite surprising is that the integer lattice \emph{is} a Lorentzian lattice in both $2+1$ and $3+1$ dimensions \cite{Schild}.
\footnote{
In $2+1$, for instance, the following boosts take the integer lattice to itself: $v_x=v_y=2/3$ and $v_x=18/35, v_y=6/7$.
}
We know from the $1+1$-dimensional integer lattice, however, that a boost along any spatial coordinate direction would create huge voids in any
higher-dimensional integer lattice. Therefore, one would expect a poor number-volume realization in this case.
We have confirmed this intuition for the $2+1$-dimensional integer lattice using simulations similar to those discussed previously (see Figure \ref{fig:LLNV-3D}).
%We have tested this in $2+1$ via simulations and found that the same problem does indeed persist. 
What makes $1+1$-dimensional Minkowski space special
is that boosts can only be performed along one direction. Therefore, a Lorentzian lattice 
does not ``change" too drastically under the action of an arbitrary boost. 
This feature does not seem to persist in higher dimensions, which leads us to conclude that Lorentzian lattices in higher dimensions are not likely to
realize the number-volume correspondence better than Poisson sprinkling.

\begin{figure}
	\centering
	\includegraphics[width=0.7\hsize]{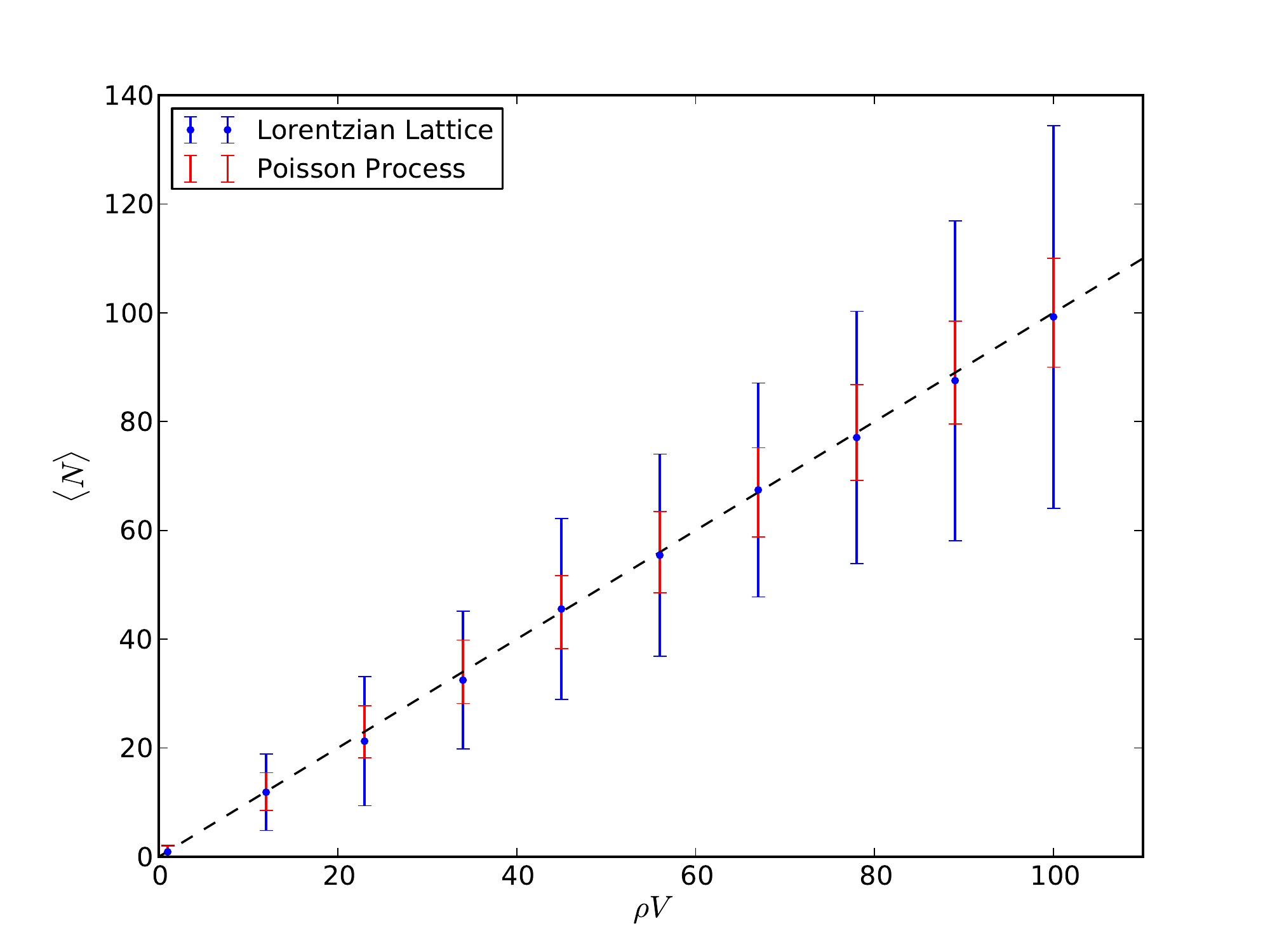}
	\caption{The number-volume correspondence for the $2+1$-dimensional integer lattice. For a given volume $V$, $200$ different causal diamonds with the same volume $V$ but randomly varying shapes are considered. The mean and standard deviation of the number of points (blue) is compared with that of the Poisson process (red).}
	\label{fig:LLNV-3D}
\end{figure}
\section{A Conjecture}
Based on the results of the previous Sections, we conjecture the following:
\begin{conj}
Let $\xi$ be a point process whose realizations are points of a $3+1$-dimensional smooth Lorentzian manifold $(M,g)$. 
Let $N_S$ be the random variable which counts the number of points in a causal interval $S\subset M$: it takes on a value $n\in\{0,1,2,\dots\}$ with probability $P_S(n)$.
%where $P_S(n)$ is the probability of realizing $n$ points in $S$.
Assume also
that $\xi$ realizes the number-volume correspondence on average $\forall$ $S$: $\langle N_S\rangle=V_S$, where
$V_S$ is the spacetime volume of $S$. 
Then, $\nexists$ $\xi$ and $V_*>0$ such that for all causal intervals $S$ with volume $V_S>V_*$:
\beq
\langle(N_S-V_S)^2\rangle\le\alpha V_S \quad
\text{where} \qquad
0\le\alpha<1.
\label{poisThm2}
\eeq
\end{conj}

%---------------------------------------------------------------------------------------------------------------------------------
%---------------------------------------------------------------------------------------------------------------------------------
%------------------------------------------------------Conclusions----------------------------------------------------------
%---------------------------------------------------------------------------------------------------------------------------------
%---------------------------------------------------------------------------------------------------------------------------------
\section{Conclusions}
Causal set theory maintains that all information about the continuum spacetime of general relativity is contained
microscopically in a partially order and locally finite set. Discreteness allows one to count elements, which is thought to
provide information about scale: a spacetime region with volume $V$ should contain about $\rho V$ causal set elements.
In this paper, we proved a theorem which shows that this number-volume correspondence is best realized via 
Poisson sprinkling for arbitrarily small volumes. Quite surprisingly, we also showed that $1+1$-dimensional Lorentzian lattices
provide a much better number-volume correspondence than Poisson sprinkling for large volumes. We presented evidence, however, that this feature should not persist 
in $3+1$ dimensions and conjectured that the Poisson process should indeed provide the best number-volume correspondence for macroscopically large spacetime regions.

%---------------------------------------------------------------------------------------------------------------------------------
%---------------------------------------------------------------------------------------------------------------------------------
%----------------------------------------------Acknowledgements-------------------------------------------------------
%---------------------------------------------------------------------------------------------------------------------------------
%---------------------------------------------------------------------------------------------------------------------------------
\acknowledgments
We are indebted to Rafael Sorkin and Niayesh Afshordi for many useful discussions throughout the course of this project. 
We thank Rafael Sorkin also for providing detailed comments on an earlier draft of our paper.
This research was supported in part by Perimeter Institute for Theoretical Physics. Research at Perimeter Institute is supported by the Government of Canada through Industry Canada and by the Province of Ontario through the Ministry of Research and Innovation.
%---------------------------------------------------------------------------------------------------------------------------------
%---------------------------------------------------------------------------------------------------------------------------------
%--------------------------------------------------------References---------------------------------------------------------
%---------------------------------------------------------------------------------------------------------------------------------
%---------------------------------------------------------------------------------------------------------------------------------

\bibliographystyle{jhep}
\bibliography{Poisson}

%---------------------------------------------------------------------------------------------------------------------------------
%---------------------------------------------------------------------------------------------------------------------------------
%----------------------------------------------------------Appendix----------------------------------------------------------
%---------------------------------------------------------------------------------------------------------------------------------
%---------------------------------------------------------------------------------------------------------------------------------
\newpage
\appendix
%---------------------------------------------------------------------------------------------------------------------------------
%-----------------------------------------------------Saravani Process----------------------------------------------------
%---------------------------------------------------------------------------------------------------------------------------------
\section{Proof of Inequality \eqref{ineQ}}
\label{saravProcess}
\begin{thm}
Let $N_V$ be a discrete random variable which takes on a value $n\in\{0,1,2,\cdots\}$ with probability $P_V(n)$, and whose mean is 
$V>0$: 
\beq
\langle N_V\rangle=\sum_{n=0}^{\infty}nP_V(n)=V.
\eeq
$N_V$ has the least variance when $P_V(n)=0$ $\forall$ $n\neq n_*,n_*+1$, where $n_*$ is the largest integer which is smaller than or equal to $V$.
Equivalently:
\beq
\langle(N_V-V)^2\rangle \ge (V-n_*)(n_*+1-V),
\eeq
where the inequality is saturated for the aforementioned process.
\end{thm} 
\begin{proof}
The following three conditions must be true %$\forall$ $V$ and $P_V$:
\begin{align}
\sum_{n=0}^{\infty}P_V(n)&=1,\label{sumOne}\\
\sum_{n=0}^{\infty}P_V(n)n&=V,\label{sumMean}\\\
0\le P_V(n)&\le1 \qquad\forall\qquad n.
\end{align}
We denote the random variable which we claim has the least variance by $N^{m}_V$, and its probability mass function
by $P^{m}_V$. 
It follows from \eqref{sumOne} and \eqref{sumMean} that
\beq
P^m_V(n_*)=n_*+1-V, \qquad
P^m_V(n_*+1)=V-n_*, \qquad
\langle(N^m_V-V)^2\rangle = (V-n_*)(n_*+1-V).
\eeq
Let us now show that for any other probability mass function $P_V(n)$:
\beq
\sigma_V^2\equiv\sum_{n=0}^{\infty}P_V(n)(n-V)^2\ge(V-n_*)(n_*+1-V).
\eeq
To this end, we define the following
\begin{align}
A_V&\equiv\sum_{n=0}^{n_*}P_V(n), \\
B_V&\equiv\sum_{n=0}^{n_*}P_V(n)(V-n)=\sum_{n=n_*+1}^{\infty}P_V(n)(n-V),
\end{align}
where the last equality follows from \eqref{sumMean}. On the one hand,
\beq
B_V=\sum_{n=0}^{n_*}P_V(n)(V-n)\ge(V-n_*)\sum_{n=0}^{n_*}P_V(n)=A_V(V-n_*).
\label{BvIneq1}
\eeq
On the other hand,
\beq
B_V=\sum_{n=n_*+1}^{\infty}P_V(n)(n-V)\ge(n_*+1-V)\sum_{n=n_*+1}^{\infty}P_V(n)=(n_*+1-V)(1-A_V).
\label{BvIneq2}
\eeq
It then follows from \eqref{BvIneq1} and \eqref{BvIneq2} that
\beq
1-\frac{B_V}{n_*+1-V} \le A_V \le \frac{B_V}{V-n_*},
\eeq
which in turn implies that
\beq
B_V\ge(V-n_*)(n_*+1-V).
\label{Bineq}
\eeq
Consider now the variance:
\begin{align}
\sigma_V^2&=\sum_{n=0}^{n_*-1}P_V(n)(n-V)^2 + \sum_{n=n_*+2}^{\infty}P_V(n)(n-V)^2\notag\\
&+P_V(n_*)(V-n_*)^2+P_V(n_*+1)(n_*+1-V)^2.
\end{align}
For all $n\neq n_*, n_*+1$, $(n-V)^2>|V-n|$, from which it follows that
\begin{align}
\sigma_V^2&\ge\sum_{n=0}^{n_*-1}P_V(n)(V-n) + \sum_{n=n_*+2}^{\infty}P_V(n)(n-V)\notag\\
&+P_V(n_*)(V-n_*)^2+P_V(n_*+1)(n_*+1-V)^2\\
&=2B_V+(n_*-V)(n_*+1-V)\left[P_V(n_*)+P_V(n_*+1)\right].
\end{align}
The equality in the last line follows from recognizing that 
\beq
\sum_{n=n_*+2}^{\infty}P_V(n)(n-V)=\sum_{n=0}^{n_*+1}P_V(n)(V-n).
\eeq
Finally, using the inequality \eqref{Bineq}:
\begin{align}
\sigma_V^2&\ge2(V-n_*)(n_*+1-V)+(n_*-V)(n_*+1-V)\left[P_V(n_*)+P_V(n_*+1)\right]\\
&=(V-n_*)(n_*+1-V)\left[2-P_V(n_*)-P_V(n_*+1)\right]\\
&\ge(V-n_*)(n_*+1-V),
\end{align}
where the last inequality follows from the fact that $P_V(n_*)+P_V(n_*+1)\le1$. This concludes the proof of the theorem.
\end{proof}
%---------------------------------------------------------------------------------------------------------------------------------
%----------------------------------------------All 2D Lorentzian Lattices------------------------------------------------
%---------------------------------------------------------------------------------------------------------------------------------
\section{2D Lorentzian Lattices: Details}
\label{LL2Ddetails}
We wish to construct a lattice that is invariant under the action of a discrete subgroup of the Lorentz group. We shall work in $D$-dimensional Minkowski space
and use the metric signature $-++\cdots$. Consider $D$ vectors $\xi_{(d)}$, with $d\in\{0,1,2,\cdots,D-1\}$, which generate the lattice. In other words, any element of the
lattice $X$ can be written as
\beq
X=n^{(d)}\xi_{(d)},
\eeq
where $n^{(d)}$ are integers and the summation over $d$ is implicit. Let $\Lambda$ be an element of the Lorentz group. We require that for all points $X$ on the lattice, $\Lambda X$ is also a point on the lattice:
\beq
\Lambda X=n^{(d)}\Lambda\xi_{(d)}=m^{(d)}\xi_{(d)},
\label{lat2lat}
\eeq
where $m^{(d)}$ are integers. We may decompose $\Lambda\xi_{(d)}$ in the basis of the generators:
\beq
\Lambda\xi_{(d)}=A_{(d)}^{\phantom{(d)}(d')}\xi_{(d')},
\label{defA}
\eeq
where $A_{(d)}^{\phantom{(d)}(d')}$ are constants which depend on $\Lambda$ and $\xi_{(d)}$.
It then follows from \eqref{lat2lat} that
\beq
n^{(d)}A_{(d)}^{\phantom{(d)}(d')}=m^{(d')}.
\eeq
\emph{
Therefore, $A_{(d)}^{\phantom{(d)}(d')}$ must be an integer for all $d$ and $d'$ if our lattice is to be invariant under the action of $\Lambda$.
}
In order to compute $A$, we can``dot" both sides of \eqref{defA} by $\xi_{(d'')}$:
\beq
\Lambda\xi_{(d)}\cdot\xi_{(d'')}=A_{(d)}^{\phantom{(d)}(d')}\xi_{(d')}\cdot\xi_{(d'')}.
\eeq
Defining the matrices $B$ and $C$ as,
\beq
B_{(d)}^{\phantom{(d)}(d')}\equiv\xi_{(d)}\cdot\xi_{(d')}, \qquad
C_{(d)}^{\phantom{(d)}(d')}\equiv\Lambda\xi_{(d)}\cdot\xi_{(d')},
\label{defBC}
\eeq
it follows that
\beq
A=CB^{-1}.
\label{Aexpr}
\eeq

Consider now the case of $1+1$ Minkowski space, i.e. $D=2$. Let $\xi_{(0)}$ and $\xi_{(1)}$ be the timelike and spacelike generators:
\beq
\xi_{(0)}=\epsilon
\begin{pmatrix}
\cosh\psi\\ 
\sinh\psi
\end{pmatrix}
, \qquad
\xi_{(1)}=\delta
\begin{pmatrix}
\sinh\theta\\ 
\cosh\theta
\end{pmatrix},
\eeq
where $\epsilon, \delta>0$. Also, since in $1+1$ we only have boosts to consider:  
\beq
\Lambda=
 \begin{pmatrix}
  \cosh\phi&  \sinh\phi \\
  \sinh\phi&  \cosh\phi
 \end{pmatrix}.
\eeq
Defining the following quantities,
\beq
\gamma=\frac{\delta}{\epsilon}, \qquad
\chi=\psi-\theta,
\label{defpigam}
\eeq
it follows from \eqref{defBC} that
\beq
B=\epsilon^2
 \begin{pmatrix}
  -1&  \gamma\sinh\chi \\
  \gamma\sinh\chi&  \gamma^2
 \end{pmatrix},
\qquad
C=\epsilon^2
 \begin{pmatrix}
  -\cosh\phi&  \gamma\sinh(\phi+\chi) \\
  \gamma\sinh(\chi-\phi) &  \gamma^2\cosh\phi
 \end{pmatrix}.
\eeq
Using \eqref{Aexpr}:
\beq
A=\frac{1}{\cosh\chi}
 \begin{pmatrix}
  \cosh(\phi+\chi) &  \frac{1}{\gamma}\sinh\phi \\
  \gamma\sinh\phi &  \cosh(\phi-\chi)
 \end{pmatrix}.
\eeq
We need to pick $\phi, \chi$ and $\gamma$ so that all elements of $A$ are integers. Let $k_1-k_4$ be integers and require
\beq
\frac{\cosh(\phi+\chi)}{\cosh{\chi}}=k_1,\qquad
\frac{1}{\gamma}\frac{\sinh\phi}{\cosh{\chi}}=k_2, \qquad
\gamma\frac{\sinh\phi}{\cosh{\chi}}=k_3, \qquad
\frac{\cosh(\phi-\chi)}{\cosh{\chi}}=k_4.
\label{2Dintg}
\eeq
Note that
\beq
k_1,k_4>0, \qquad
\text{sgn}(k_2)=\text{sgn}(k_3).
\eeq
The second and third equations in \eqref{2Dintg} are equivalent to
\beq
\gamma^2=\frac{k_3}{k_2},\qquad
\frac{\sinh^2\phi}{\cosh^2\chi}=k_2k_3.
\label{equiv2D1}
\eeq
Also, the first and fourth equations in \eqref{2Dintg} imply
\beq
2\cosh\phi=k_1+k_4,\qquad
2\sinh\phi\tanh\chi=k_1-k_4.
\label{equiv2D2}
\eeq
The first equation in \eqref{equiv2D2} fixes $\phi$ up to a sign, using which the second equation in 
\eqref{equiv2D1} fixes $\chi$ up to a sign. Putting these together in the second equation in \eqref{equiv2D2},
we obtain the following constraint on the integers $k_1-k_4$:
\beq
k_1k_4-k_2k_3=1.
\eeq
This equation can be satisfied for various integers, and therefore there are many Lorentzian lattices in $1+1$. 

\emph{To summarize}: find integers $k_1-k_4$ that satisfy the conditions $(i)$ $k_1,k_4>0$, $(ii)$ sgn($k_2$)=sgn($k_3$), (iii) $k_1k_4-k_2k_3=1$. Then, if we let
$\cosh(\phi)=\frac{k_1+k_4}{2}$, $\gamma=\sqrt{k_3/k_2}$, and $\sinh(\chi)=\frac{k_1-k_4}{2\sqrt{k_2k_3}}$, the lattice generated by
$\xi_{(0)}$ and $\xi_{(1)}$ goes to itself under the action of $\Lambda(\phi)$, with $\psi,\theta,\delta,\epsilon$ satisfying \eqref{defpigam}.
Figure \ref{fig:lorlat} shows an example of a Lorentzian lattice with $k_1=2$, $k_2=k_3=k_4=1$, $\delta=1$, and $\theta=0$.

\end{document}